\documentclass[11pt, oneside]{article}   	
\usepackage{geometry}                		
\usepackage{graphicx,caption}				

\usepackage{hyperref}

\usepackage{amssymb}
\usepackage{amsmath}
\usepackage{algorithmic}
\usepackage{algorithm}
\usepackage{color}
\usepackage{amsthm}
\usepackage{bbm}
\usepackage{soul,xcolor}
\usepackage{booktabs}
\usepackage{multirow}
\usepackage{adjustbox}

\usepackage{textcomp}
\usepackage[shortlabels]{enumitem}
\usepackage[caption=true, position=bottom, labelformat=empty]{subfig}
\setstcolor{red}

\newtheorem{thm}{Theorem}[section]

\newtheorem{prop}[thm]{Proposition}

\newtheorem{defn}{Definition}[section]

\newtheorem{rem}{Remark}

\usepackage{authblk}

\DeclareMathOperator*{\argmin}{arg\,min}

\title{An Optimal Control Strategy for Execution of Large Stock Orders Using LSTMs\thanks{This work was partially supported by NSF grant DMS-1907518 and in part by the New York University Abu Dhabi (NYUAD) Center for
Artificial Intelligence and Robotics, funded by Tamkeen under the NYUAD Research Institute Award CG010.\\We are grateful to the Associate Editor and anonymous referees, each of whom provided invaluable feedback on this work.}}

\author{A. Papanicolaou,\thanks{Department of Mathematics, North Carolina State University, Raleigh, NC 27695-8205, {\tt\small apapani@ncsu.edu}}~~~H. Fu,\thanks{Department of Electrical and Computer Engineering, NYU Tandon School of Engineering, Brooklyn, NY 11201, {\tt\small hf881@nyu.edu}}~~~P. Krishnamurthy,\thanks{Department of Electrical and Computer Engineering, NYU Tandon School of Engineering, Brooklyn, NY 11201, {\tt\small prashanth.krishnamurthy@nyu.edu}}~~~B. Healy,\thanks{Institute of Finance and Technology, UCL, London, WC1E 6BT,  {\tt\small b.healy@ucl.ac.uk}}~~~F. Khorrami\thanks{Department of Electrical and Computer Engineering, NYU Tandon School of Engineering, Brooklyn, NY 11201, {\tt\small khorrami@nyu.edu}}}

\date{}							

\begin{document}
\maketitle
\begin{abstract}
In this paper, we simulate the execution of a large stock order with real data and general power law in the Almgren and Chriss model. The example that we consider is the liquidation of a large position executed over the course of a single trading day in a limit order book. Transaction costs are incurred because large orders walk the order book, that is, they consume order book liquidity beyond the best bid/ask. We model the order book with a power law that is proportional to trading volume, and thus transaction costs are inversely proportional to a power of trading volume. We obtain a policy approximation by training a long short term memory (LSTM) neural network to minimize transaction costs accumulated when execution is carried out as a sequence of smaller suborders. Using historical S\&P100 price and volume data, we evaluate our LSTM strategy relative to strategies based on time-weighted average price (TWAP) and volume-weighted average price (VWAP). For execution of a single stock, the input to the LSTM is the cross section of data on all 100 stocks, including prices, volumes, TWAPs and VWAPs. By using this data cross section, the LSTM should be able to exploit inter-stock co-dependence in volume and price movements, thereby reducing transaction costs for the day. Our tests on S\&P100 data demonstrate that in fact this is so, as our LSTM strategy consistently outperforms TWAP and VWAP-based strategies.\\
\newline
\textbf{Keywords:} Price Impact, Order Books, Optimal Execution, LSTM Networks.\\
\textbf{JEL Codes:} C45, G10
\end{abstract}

\section{Introduction}
Institutional investors must consider transaction costs when trading large amounts of stock. For example, each month a multibillion dollar mutual fund may execute several large stock orders when they rebalance their stock holdings. In the U.S. stock market, a large order would be to sell 1,000,000 shares of a stock with a typical daily trading volume of 25 million shares. A naive strategy is to place a single, very large market sell order on the exchange. This single order will consume so much of the available liquidity in the limit order book that there will be an average price per share that is considerably lower than the best bid -- if the order gets filled at all. A better strategy is to divide the trade into smaller suborders, which then get executed over the course of a fixed time period. In this paper, we cast this subdividing of large trades as an optimal control problem, and then we obtain an optimal policy for execution by training a long short term memory (LSTM) neural network \cite{hochreiter1997long} to minimize losses to transaction costs.

Stock market liquidity is made available by market makers who submit limit orders at the different price ticks in the order book. Liquidity is consumed when a trader submits a market order to buy or sell. A market order that is not too big will get filled by limit orders at the best available prices. A large market buy (sell) order will \textit{walk the order book}, that is, it will consume all liquidity at multiple ticks. Walking the order book results in an average price per share that is equal to the initial best ask (bid) plus (minus) a transaction cost. Statistical studies of order book data have shown that the depth to which a large order walks the order book is approximately a concave power law of the number of shares \cite{almgren2005direct,weber2005order}. Simple calculation will show that subdividing a large order into a sequence of suborders will reduce these transaction costs, but to optimally subdivide is more complicated because there are several (stochastic) variables to consider when designing a policy.

It is common practice to evaluate a policy in terms of the average price over the entire trade. Average prices to consider are the time-weighted average price (TWAP) and the volume-weighted average price (VWAP). In general, an optimal suborder policy will minimize the expected value of transaction costs. Strategies that aim to achieve TWAP or VWAP can be optimal for executing large orders \cite{cartea2016closed,kato2015vwap}.

Our approach to the design of an optimal execution policy is to consider the problem amidst the uncertainty of real market data. Machine learning and deep neural networks are very good at learning policies directly from data without the assistance of models. For large-cap U.S. stocks (e.g., the S\&P100 constituents), there is plenty of data available upon which to train a network and conduct backtests. In this paper we perform these backtests, the conclusions from which are that neural networks can be effective for improvement of large-order execution strategies. Our implementation uses LSTM networks, which are a good candidate for constructing a policy function because they (a) do not require a model for the distribution of the market; we can learn directly from the data, (b) can handle Markov or non-Markov states, (c) are well-suited to learn in the episodic environment of single-day execution, and (d) can provide a policy with the required temporal dynamic. 

Of particular relevance in this paper is the possible presence of inter-stock co-dependence in price and volume movements. The execution problem is posed for a single stock, but the input to the network includes information from all stocks in the dataset, which allows the LSTM network to learn inter-stock dependencies that may provide better prediction of prices and volumes, thereby improving the execution prices achieved by the policy. Inputting such a large amount of data might be problematic for more parsimonious policy functions, but the LSTM is more than capable of handling this high dimensionality, and indeed, from our results on S\&P100 stocks it appears that LSTM does learn useful patterns in this dataset.

\subsection{Background and Review}
A prototypical model for optimal execution was introduced in \cite{almgren2001optimal}. A simplification of their paper is to assume that liquidity consumed by a market order refills quickly in the time between scheduled suborders, and therefore a simple formulation of the problem will consider only the temporary impact on price. The speed of impacted-price reversion is studied in \cite{huberman2005optimal} with discussion on how faster reversion rates affect the aggressiveness of trade execution. There are also models that consider price impact with slower reversion of impacted price, for example \cite{alfonsi2016dynamic,amaral2019price,bacry2015hawkes} using Hawkes processes, and \cite{obizhaeva2013optimal} where it is shown that an optimal execution policy should start with a large order to disrupt the order book's supply and demand balance, and then begin trading continuously as the order book refills. 

Power laws have been observed in stock markets in the UK, China and the United States, as reported and discussed in \cite{gould2013limit,bouchaud2002statistical, potters2003more}, \cite{zovko2002power}, \cite{maskawa2007correlation}, \cite{gu2008empirical}. \cite{almgren2005direct} estimated the power-law exponent to be around 0.67, which they obtained from a large dataset from Citibank. Closely related to our paper is \cite{hendricks2014reinforcement} where a reinforcement approach is used to approximate the solution to \cite{almgren2001optimal}, and also the Deep Q-Learning approach to optimal execution that was studied in \cite{ning2021double}. 

Recently, there has been some progress on the development of machine learning algorithms for order book modeling and price execution \cite{lin2020end,nevmyvaka2006reinforcement,zhang2019deeplob,ZGA21}, and also \cite{schnaubelt2022deep} where reinforcement learning is implemented for optimal limit order placement in crypto markets.

\subsection{Results in this Paper}
The main result in this paper is the improved execution strategies that we find by using LSTM. We assume that limit order depth at each tick is proportional to volume and increases by a power law across ticks as we move farther from the best bid/ask. A major advantage of our LSTM approach is that the network's input includes the cross section of market data, thereby utilizing any inter-stock co-dependence that may be present in volume or price changes. To evaluate the efficacy of our approach, we implement LSTM execution on historical minute-by-minute stock market data from January 2020 - July 2022. Our results indicate that, compared with TWAP and VWAP strategies, execution with a trained LSTM network can save a 1-2 basis points (bps) per stock on a given day when executing a block trade of S\&P100 stocks.

\section{Order Book Model and Optimal Policies}
Let $S_t$ and $V_t$ denote the mid price and volume, respectively, of a stock at time $t$. Following the model described in \cite{platania2018modelling,rogers2010cost}, the order book has limit order distribution $\rho(t,s)\geq 0$, where the units of $s$ are ticks relative to $S_t$. Ticks with limit sell orders correspond to $s>1$, ticks with $s\in(0,1)$ correspond to limit buy orders, and the mid price corresponds to $s=1$. An order of $a$ shares consumes liquidity up to a relative price $r_t(a)$ such that
\begin{equation}
    \label{eq:tick_depth}
    a = \int_1^{r_t(a)}\rho(t,s)ds\ .
\end{equation}
A simple form for $\rho$ has limit orders distributed continuously in $s$, proportionally to volume with the following power law,
\begin{equation}
    \label{eq:rho}
    \rho(t,s) = \frac{V_t}{\epsilon }|s-1|^\beta 
\end{equation}
where $\beta\in[0,\infty)$, $\epsilon>0$ is a scaling parameter, and $V_t$ is the trading volume at time $t$. In this paper, we seek to optimize execution of a large order over the course of a single trading day, in which case each $V_t$ will be the total volume of trades that occurred in the $t^{th}$ minute. The impact function in \eqref{eq:rho} is like the power law considered in \cite{qb20201008,qb20200424}.

When relative price $r_t(a)$ in \eqref{eq:tick_depth} is computed with distribution $\rho(t,s)$ in \eqref{eq:rho}, we see a price that is a concave function of order size divided by volume,
\begin{equation}
    \label{eq:impactedPrice}
    r_t(a) = 1+\hbox{sign}(a)\left(\frac{\epsilon(\beta+1)}{V_t}|a|\right)^{\tfrac{1}{\beta+1}}\ .
\end{equation}  
The quantity $S_t r_t(a)$ can be thought of as the \textit{impacted price}. If impacted price is assumed to be a linear function of $a$, then the implication is that $\beta=0$ so that the order book has equal liquidity at all ticks. The prevailing conclusion in many empirical studies is that impacted price is a sub-linear concave function \cite{almgren2005direct,bershova2013non,bouchaud2010price,cont2014price}, such as the square root function corresponding to the case $\beta =1$. Cases where $\beta$ is less than zero are not considered because this would imply decreasing liquidity in successive ticks beyond the best bid/ask, which is rarely the case for liquid, large-cap stocks. 

The transaction costs incurred by walking the order book, as  described by \eqref{eq:tick_depth}, \eqref{eq:rho} and \eqref{eq:impactedPrice}, will be a convex function of trade size. The dollar amount of trading loss due to the price impact is computed as follows,
\begin{align}
    \label{eq:loss}
    \hbox{loss}(t,a)&=
    S_t\left|\int_1^{r_t(a)}s\rho(t,s)ds - a\right|=C_{\epsilon,\beta}S_t(V_t)^{-\tfrac{1}{\beta+1}}|a|^{\tfrac{\beta+2}{\beta+1}}\ ,
\end{align}
where $C_{\epsilon,\beta} =\tfrac{1}{\beta+2} (\epsilon(\beta+1))^{\tfrac{\beta+2}{\beta+1}}$. From the convexity of \eqref{eq:loss} with respect to $|a|$ it is clear that very large orders should be divided into suborders to reduce transaction costs.

\begin{rem}[Cost of Paying the Spread]
This order book model ignores the bid-ask spread. For liquid stocks the bid-ask spread is usually 1 tick, or equivalently 1\textcent.
For such stocks the transaction costs for market orders filled at the best bid/ask can be proxied by $\$.005$. This amounts to a flat fee for execution of an order and will remain constant for all the strategies that we test in this paper. Therefore, we omit the cost of the spread.
\end{rem}

\begin{rem}[Exchange Fees]
    Typically there are exchange fees that may be proportional to the dollar amount traded for smaller orders. In this paper we do not consider these fees because the problem we are considering is from the perspective of a large institutional investor for whom fees are structured differently, usually decreasing percentage-wise as the trade size increases.
\end{rem}

\begin{rem}[Permanent Impact]
    We do not consider permanent impact in this paper. The assumption is that we are trading in highly liquid stocks for which the order book is replenished very quickly after a suborder. It would certainly be interesting to consider execution with permanent price impact, but in this paper we focus our effort on finding policies that are able to optimize amidst the stochasticity and uncertainty in real-life historical price and volume data. 
\end{rem}

\subsection{Optimal Execution Policy}
\label{sec:optimalPolicy}
Let us work on a probability space $(\Omega,\mathcal F,(\mathcal F_t)_{t=0,1,2,\dots},\mathbb P)$ where $\mathcal F_t$ denotes the $\sigma$-algebra representing all information known to us at time $t$. Assume that times $t=0,1,2,3,\dots$ are equally spaced. Let the initial inventory be a number of shares $x$, and consider the situation where this inventory needs to be completely liquidated by terminal time $T$; in the example we present the times $t=0$ and $t=T$ are the opening and closing of the trading day, respectively. Let $X_t$ denote the number of remaining unexecuted shares at time $t$, for $t=1,2,3,\dots, T $. Initially we have $X_0=x$. An execution policy is a sequence of $\mathcal F_{t}$-adapted suborders $a_t$ such that $X_t = X_{t-1}+a_{t-1}$, that is, $a_{t-1}$ is this execution policy's suborder placed at time $t-1$ and executed at time $t$; these suborders are chosen so that $X_T=0$. Using the loss function given in \eqref{eq:loss}, an optimal execution policy is the minimizer of the expected loss,
\begin{align}
    \label{eq:optimal_exec}
    \min_a~& \mathbb E\sum_{t=1}^T\hbox{loss}(t,a_{t-1})\\
    \nonumber
    &\hspace{.2cm}\hbox{s.t.}\\
    \nonumber
    X_{t} &= X_{t-1} + a_{t-1}\\ 
    \nonumber
    X_T &= 0\\
    \nonumber
    X_0&=x\ .
\end{align}
where the minimization is carried out over the family of $\mathcal F_t$-measurable policies $a_t$. In \eqref{eq:optimal_exec} we are allowing for broad generality of the processes $(S_t,V_t)$, aside from them being well defined on probability space $(\Omega,\mathcal F,(\mathcal F_t)_{t=0,1,2,\dots},\mathbb P)$, and also assuming that our trading does not affect them. Later on we will narrow the family of policies to those $a_t$ that are given by an LSTM network, but still this narrowing of policies will permit broad generality in the distribution of $(S_t,V_t)$ such as being non-Markovian and having nonlinear dependence in volume.

The loss function in \eqref{eq:optimal_exec} is a risk-neutral optimization in the sense that there is no penalty on variance or risk. The optimization in \eqref{eq:optimal_exec} results in a policy that is on the frontier constructed in \cite{almgren2001optimal}; their frontier is formed by minimizing implementation shortfall with a penalty on its variance. The objective in \eqref{eq:optimal_exec} is similar to the objective in \cite{kato2015vwap}, which is risk neutral and also volume dependent.

\subsection{TWAP and VWAP Strategies}
\label{sec:TWAPandVWAP}
Before we approach solving \eqref{eq:optimal_exec} with full generality in $(S_t,V_t)$, we first discuss the two industry-standard benchmarks for large order execution, namely TWAP and VWAP, and how they are related to \eqref{eq:optimal_exec}.

\begin{defn}[TWAP]
    \label{def:TWAP}
    The time weighted average price (TWAP) is $\overline{S_T} = \frac{1}{T}\sum_{t=1}^TS_t$.
\end{defn}
\noindent The TWAP is a target for some execution policies because average execution price for large orders are often benchmarked against TWAP. A common execution strategy is the so-called \textit{TWAP strategy}, wherein the policy is to subdivide the order into equally-sized deterministic suborders, 
\begin{equation}
    \label{eq:TWAPstrat}
    a_t = -\frac{x}{T}\qquad\hbox{for }t=0,1,\dots,T-1\ .
\end{equation}
The TWAP is easy to implement, as it is assured to satisfy terminal condition $X_T=0$ and doesn't require any parameter estimation. However, the TWAP strategy ignores volume and any other pertinent information acquired during the trading period. Indeed, volume is often a concern when it comes to evaluating trades, which is why execution policies often target the VWAP (see \cite{cartea2016closed} for more on VWAP targeting). 
\begin{defn}[VWAP]
    \label{def:VWAP}
    The volume weighted average price (VWAP) is $\overline{S_V} = \frac{\sum_{t=1}^TV_tS_t}{\sum_{t=1}^TV_t}$.
\end{defn}
\noindent A VWAP strategy is to subdivide proportional to moments of the volume,
\begin{align}
    \label{eq:VWAPstrat}
    a_t &= -\frac{x\overline V_{t+1}}{\sum_{t=1}^T\overline V_t}\qquad\hbox{for }t=0,1,2\dots,T-1\ ,
\end{align}
where $\overline{V_t}=\left(\mathbb E(V_t)^{-\frac{1}{\beta+1}}\right)^{-(\beta+1)}$ (see \cite{kato2015vwap}). This VWAP strategy can be effective for single-day execution because volume follows a somewhat predictable ``U" shape (see Figure \ref{fig:Ushape}). In practice, the building of a VWAP strategy requires some prior data to determine the typical evolution of volume over a trading period. For example, if we observe a history of volumes from past trading days then we can use historically estimated $\overline {V_t}$'s in \eqref{eq:VWAPstrat}.

\begin{figure}[ht]
   \centering
    \captionsetup{width=.8\linewidth}
     \includegraphics[scale=0.5]{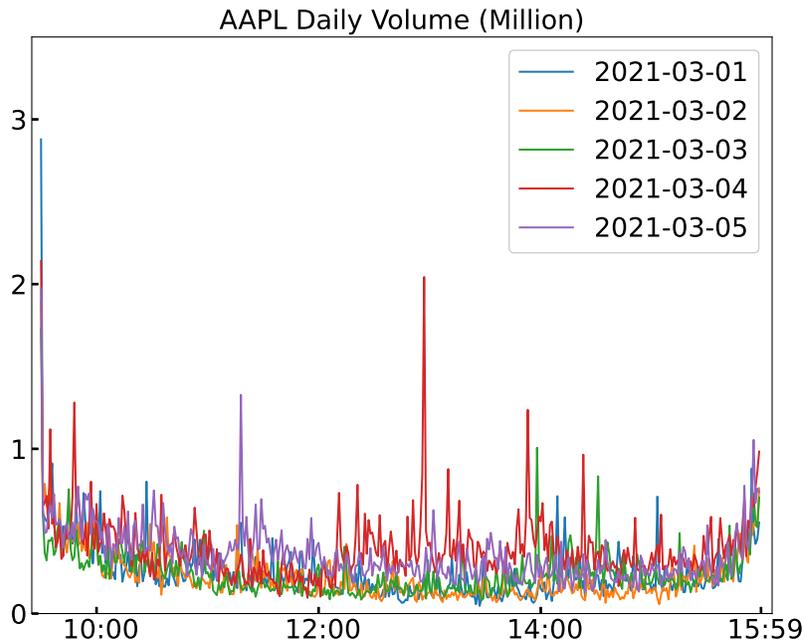}
     \caption{The ``U"-shaped pattern of daily volume in Apple.}
\label{fig:Ushape}
\end{figure}

In general, the TWAP strategy of \eqref{eq:TWAPstrat} and the VWAP strategy of \eqref{eq:VWAPstrat} do not minimize the loss in \eqref{eq:optimal_exec}. However, with some simplification of the volume process and minimal assumptions on $S_t$, we can show that VWAP is optimal if $V_t$ is a deterministic function of $t$, and we can show TWAP is optimal if $V_t$ is deterministic and constant in $t$.  

\begin{prop}[Deterministic Volume]
    \label{prop:optimal_deterministic}
    Suppose $V_t=\mathbb EV_t=\overline {V_t}$ for all $t$, that is, volume is a deterministic function of $t$. Also assume that $S_t$ is a martingale with respect to the filtration $(\mathcal F_t)_{t=0,1,2,\dots}$. Then the VWAP strategy \eqref{eq:VWAPstrat} is optimal, and if $\overline{V_t}$ is constant in $t$ then the TWAP strategy of \eqref{eq:TWAPstrat} is optimal. 
\end{prop}
\begin{proof}
    For deterministic volume and omitting constant $C_{\epsilon,\beta}$, the optimization in \eqref{eq:optimal_exec} for stock $i$ can be posed as
\begin{align*}
    \min_a~&\mathbb E\Bigg[\sum_{t=1}^TS_t(\overline{V_t})^{-\tfrac{1}{\beta+1}}|a_{t-1}|^{\tfrac{\beta+2}{\beta+1}} \Bigg]\\
    &\hspace{.2cm}\hbox{s.t.}\\
    \nonumber
    X_{t} &= X_{t-1} + a_{t-1}\\ 
    X_T &= 0\\
    X_0&=x\ .
\end{align*}
This optimization can be written as a Lagrangian,
\begin{align*}
    &\mathbb E\sum_{t=0}^{T-1}\left(S_{t+1}(\overline{V}_{t+1})^{-\tfrac{1}{\beta+1}}|a_{t}|^{\tfrac{\beta+2}{\beta+1}} +\lambda_t(X_{t+1}-X_t-a_{t})\right)\ ,\\
    &=\mathbb E\left[\sum_{t=0}^{T-1}\left(S_{t+1}(\overline{V}_{t+1})^{-\tfrac{1}{\beta+1}}|a_{t}|^{\tfrac{\beta+2}{\beta+1}} -(\lambda_{t+1}-\lambda_t)X_{t+1}-\lambda_ta_{t}\right)+\lambda_TX_T-\lambda_0X_0\right]\ ,
\end{align*}
with terminal condition $X_T=0$, initial condition $X_0=x$, and where $\lambda_{t}$ is an $\mathcal F_{t}$-adapted Lagrange multiplier process. First order conditions in $a_{t}$ and $X_{t+1}$ yield the following co-state equations,
\begin{align*}
    \frac{\beta+2}{\beta+1}\mathbb E_{t}\left[\hbox{sign}(a_{t})S_{t+1}\left(\frac{|a_{t}|}{\overline{V}_{t+1}}\right)^{\frac{1}{\beta+1}}\right] -\lambda_{t} &= 0\quad\hbox{for $t=0,1,2\dots,T-1$}\\
    \lambda_{t}-\mathbb E_{t}\lambda_{t+1}&=0\quad\hbox{for $t=0,1,2\dots,T-2$}\ ,
\end{align*}
with $\mathbb E_{t}$ denoting expectation conditional on $\mathcal F_{t}$. For $x>0$ the optimal policy is the VWAP strategy,
\begin{align*}
    a_t&= -\frac{x\overline{V}_{t+1}}{\sum_{t=1}^T\overline{V_t}}\\
    \lambda_{t}&=-\frac{\beta+2}{\beta+1}\left(\frac{x}{\sum_{t=1}^T\overline{V_t}}\right)^{\frac{1}{\beta+1}}\mathbb E_{t}S_{t+1}\ ,
\end{align*}
where the martingale property $S_t=\mathbb E_{t}S_{t+1}$ ensures that $\lambda_{t}=\mathbb E_{t}\lambda_{t+1}$. Furthermore, this optimal policy is the TWAP strategy if $\overline{V_t}$ is constant in $t$. The analogous proof holds for $x<0$. 
\end{proof}

In \cite{kato2015vwap} there are results similar to Proposition \ref{prop:optimal_deterministic}, and also a theorem suggesting that under certain Markovian assumptions, a deterministic VWAP is the optimal $\mathcal F_t$-adapted policy for stochastic volume. The following proposition shows how the VWAP strategy in \ref{eq:VWAPstrat} can be optimal under certain assumptions on $S_t$ and $V_t$.

\begin{prop}[Stochastic Volume]
    \label{prop:optimal_stochastic}
     Define
    \[M_t = \frac{\left(\mathbb EV_t^{-\frac{1}{\beta+1}}\right)^{-1}}{V_t^{\frac{1}{\beta+1}}}\ ,\] 
    and assume $M_t$ is a martingale with respect to the filtration $(\mathcal F_t)_{t=0,1,2,\dots}$. Also assume that $S_t$ is a martingale with respect to the filtration $(\mathcal F_t)_{t=0,1,2,\dots}$, independent of $M_t$. Then the VWAP strategy \eqref{eq:VWAPstrat} is optimal. 
\end{prop}

\begin{proof}
    Taking the Lagrangian approach as we did in the proof of Proposition \ref{prop:optimal_deterministic}, we arrive at the following co-state equations,
    \begin{align*}
    \frac{\beta+2}{\beta+1}\mathbb E_{t}\left[\hbox{sign}(a_{t})S_{t+1}\left(\frac{|a_{t}|}{V_{t+1}^i}\right)^{\frac{1}{\beta+1}}\right] -\lambda_{t} &= 0\\
    \lambda_{t}-\mathbb E_{t}\lambda_{t+1}&=0\ .
\end{align*}
For $x>0$ we make the ansatz that $a_t<0$, which gives us
\[\lambda_t=-\frac{\beta+2}{\beta+1}\mathbb E_{t}\left[S_{t+1}\left(\frac{|a_{t}|}{V_{t+1}}\right)^{\frac{1}{\beta+1}}\right]\ ,\]
and if we insert the VWAP strategy of \eqref{eq:VWAPstrat} we see that
\begin{align*}
    &\mathbb E_{t}\left[S_{t+1}\left(\frac{|a_{t}|}{V_{t+1}}\right)^{\frac{1}{\beta+1}}\right]\\
    &=S_t\mathbb E_{t}\left[\left(\frac{|a_{t}|}{V_{t+1}}\right)^{\frac{1}{\beta+1}}\right]\\
    &=\left(\frac{x}{\mathcal K}\right)^{\frac{1}{\beta+1}}S_t\mathbb E_{t}\left[\left(\frac{\left(\mathbb E(V_{t+1})^{-\frac{1}{\beta+1}}\right)^{-(\beta+1)}}{V_{t+1}}\right)^{\frac{1}{\beta+1}}\right]\\
    & = \left(\frac{x}{\mathcal K}\right)^{\frac{1}{\beta+1}}S_t\mathbb E_tM_{t+1}\\
    &=\left(\frac{x}{\mathcal K}\right)^{\frac{1}{\beta+1}}S_tM_t\ ,
\end{align*}
where $\mathcal K$ is the denominator of $a_t$ in \eqref{eq:VWAPstrat}. Thus, when VWAP strategy \eqref{eq:VWAPstrat} is used we have $\lambda_t=-\left(\frac{x}{\mathcal K}\right)^{\frac{1}{\beta+1}}\frac{\beta+2}{\beta+1}S_tM_t=-\left(\frac{x}{\mathcal K}\right)^{\frac{1}{\beta+1}}\frac{\beta+2}{\beta+1}\mathbb E_tS_{t+1}M_{t+1}=\mathbb E_t\lambda_{t+1}$, thereby confirming that it is an optimal policy. The analogous proof holds for $x<0$.
\end{proof}
An example for martingale $M_t$ in Proposition \ref{prop:optimal_stochastic} is log-normal volume (same as the example given in \cite{kato2015vwap}), $\log(V_{t+1}/V_t) = \mu_t + \sigma_tZ_{t+1}$, where $(Z_t)_{t=1,2,\dots}$ is a sequence of standard normals independent of the past, and where $\mu_t$ and $\sigma_t$ are deterministic functions of $t$ that are calibrated so that $V_t$ adheres to a U-shape over the course of a trading day; it is straightforward to check if parameters $\mu_t$ and $\sigma_t$ allow for $M_t$ to be a martingale. At this point however, rather than pursue specification of the underlying stochastic processes, we instead choose to leave the data distribution unspecified and then proceed to train an LSTM to trade optimally based on historical observations of financial data.
 
\section{LSTM Execution Policy -- Experimental Setup}

A policy approximation for the optimal solution of \eqref{eq:optimal_exec} can be obtained by training a neural network on historical market data. Our approach is to train an LSTM neural network to minimize the objective in \eqref{eq:optimal_exec}, and then compare with the TWAP strategy of \eqref{eq:TWAPstrat} and the VWAP strategy of \eqref{eq:VWAPstrat}. The choice of LSTM rather than a convolutional neural network (CNN) or a recurrent neural network (RNN) is based on the following two considerations. Firstly, the problem in \eqref{eq:optimal_exec} is time-dependent, requiring the neural network to memorize prior information. A CNN is static and thus cannot memorize prior information for use in backpropagation whereas a RNN can memorize prior information, but suffers from the gradient vanishing problem \cite{hochreiter1998vanishing}. The LSTM is able to handle both of these considerations.

The backtests we conduct will have a fixed number of suborders and a fixed submission time. We will submit suborders every 5 minutes for a total of 78 total suborders over the 390 minutes of the trading day. Each of the strategies we test, namely the TWAP, VWAP and LSTM strategies, will execute in the same 5 minute intervals, thus ensuring a fair comparison. To be as realistic as possible, we also assume that the LSTM lags one minute behind the real-time market (i.e., the input of the LSTM network only includes the data up to and including the prior minute when the suborder is executed) to prevent the LSTM from having any foresight bias.

\begin{figure}
    \centering
    \includegraphics[scale=0.5]{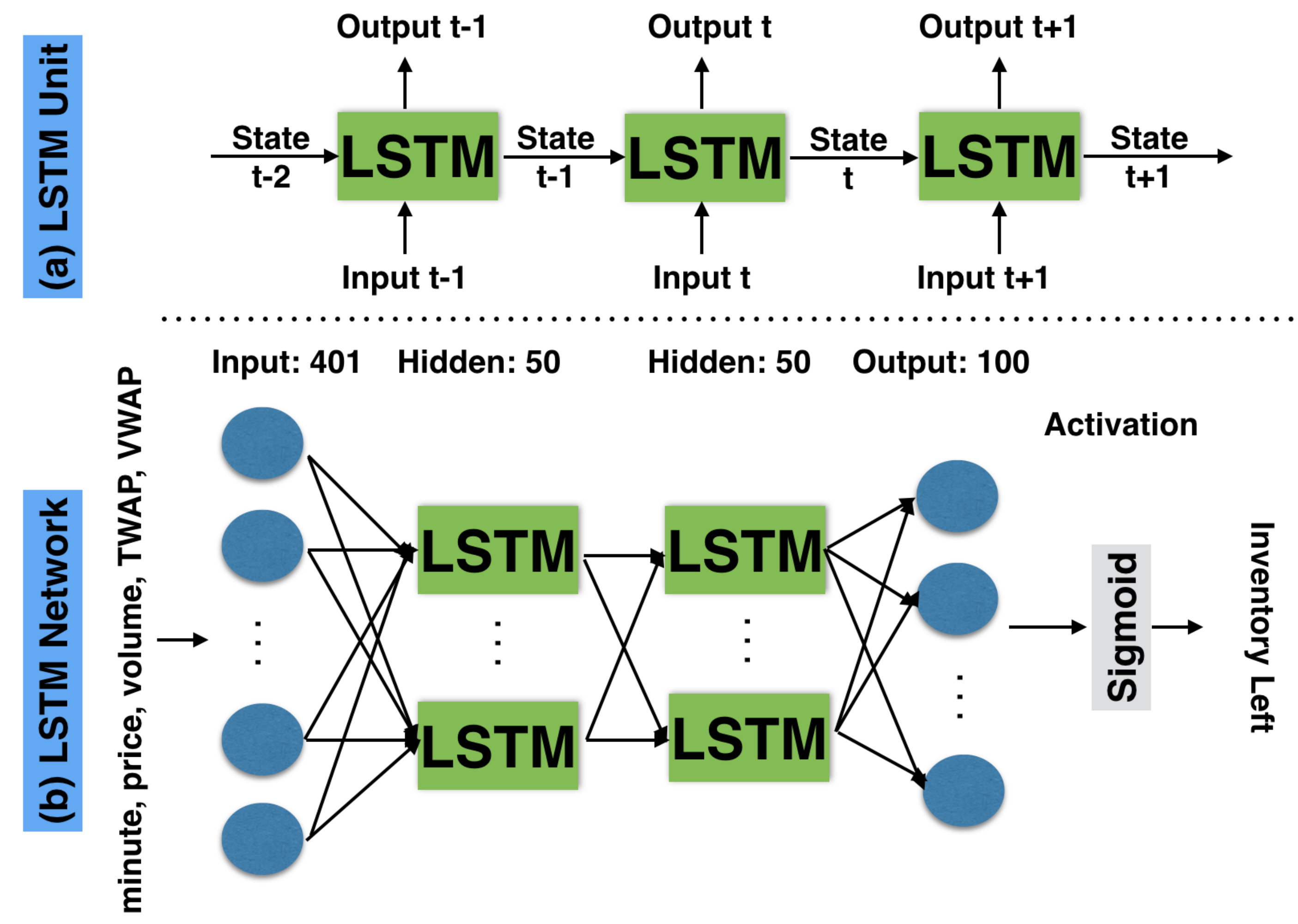}
    \caption{(a) The structure of LSTM unit. (b) The structure of the used LSTM network.}
    \label{fig:structure}
\end{figure}

Fig.~\ref{fig:structure} (a) shows the structure of a single LSTM unit. We refer to the internal parameters of the LSTM unit as the state.  At a particular time $t$, the LSTM unit updates its state to $state_t$  using the old state, $state_{t-1}$, and the new input, $input_t$. The LSTM unit also generates an output, $output_t$. While $state_t$ is then used to update $state_{t+1}$, $output_t$ is used for other calculation purposes. Fig.~\ref{fig:structure} (b) shows the structure of the LSTM network used: it has two LSTM layers with 50 LSTM units in each. With fewer or smaller LSTM layers the network tends to underfit, whereas larger,  more complex LSTMs have comparable performance to that of our architecture but with increased computational complexity. The input has length 401, which comprises the current minute, the S\&P 100 stocks' prices, their volumes, and the remaining inventories in each stock under both TWAP and VWAP strategies. By including in the input, the current inventories remaining under TWAP and VWAP strategies, the LSTM strategy performs no worse than TWAP or VWAP as the LSTM can simply replicate them. The LSTM state also contains the current level of inventory in stock $i$, and after passing the results through a sigmoid activation function, the network outputs the updated inventory remaining in each stock $i$ for the current time. The total number of network parameters is 116,100.

\subsection{Data and Parameter Estimation}

Our data was obtained from FirstRate Data and consists of minute-by-minute prices and volumes from January 2, 2020, through July 1, 2022, for the S\&P100 stocks as listed on March 21, 2022. The data was split into nine groups for training and testing the LSTM networks, as shown in Table~\ref{tab:dataset}. The S\&P100 contains the largest 100 companies in the U.S. stock market by market capitalization. The limit order book for each of these stocks is extremely deep at all times, meaning that there is plenty of liquidity and it is very unlikely that a suborder will not get filled. For these liquid stocks, the simulated performance of the LSTM policy will be a realistic characterization of how it will perform in real-life trading.

\begin{table}
    \centering
    \caption{The datasets and their sizes.}
    \begin{tabular}{|c|cc|cc|}
    \hline
        Fold & Training Data & Days & Testing Data  & Days \\
         \hline
       1 & From 2020-01-02 To 2020-03-27 & 60 & From 2020-03-30 To 2020-06-03 & 45 \\
       2 & From 2020-03-30 To 2020-06-23  & 60 & From 2020-06-24 To 2020-08-27 & 45 \\
       3 & From 2020-06-24 To 2020-09-17  & 60 & From 2020-09-18 To 2020-11-20 & 45 \\
       4 & From 2020-09-18 To 2020-12-11  & 60 & From 2020-12-14 To 2021-02-19 & 45 \\
       5 & From 2021-01-04 To 2021-03-30  & 60 & From 2021-03-31 To 2021-06-04 & 45 \\
       6 & From 2021-03-31 To 2021-06-24  & 60 & From 2021-06-25 To 2021-08-30 & 45 \\
       7 & From 2021-06-25 To 2021-09-20  & 60 & From 2021-09-21 To 2022-11-23 & 45 \\
       8 & From 2021-09-21 To 2021-12-14  & 60 & From 2021-12-15 To 2022-02-18 & 45 \\
       9 & From 2022-01-03 To 2022-03-29  & 60 & From 2022-03-30 To 2022-06-03 & 45 \\
       \hline
    \end{tabular}
    \label{tab:dataset}
\end{table}

For the power law in \eqref{eq:rho}, we take $\beta=.67$ so that the impact in \eqref{eq:impactedPrice} has a power of $.6$ as suggested by \cite{almgren2005direct}. However, we will conduct backtests both with constant $\beta$ and stochastically fluctuating $\beta$, the latter being a more realistic description of real-life order books. We set $\epsilon$ so that the transaction cost equals 0.01\%-0.02\% (i.e., 1-2 bps) of the value traded, which is realistic for S\&P100 stocks. For example, when trading 1 million shares, an appropriate $\epsilon$ would be 0.003. Then, $C_{\epsilon, \beta}$ is calculated by:
\begin{align}
\label{eq:liquid_ratio}
    C_{\epsilon, \beta} = \tfrac{1}{\beta+2} (\epsilon(\beta+1))^{\tfrac{\beta+2}{\beta+1}} \approx 7.87\times 10^{-5}.
\end{align}

For $i=1,2,\dots,100$ let $S_t^i$ and $V_t^i$ denote the price and volume for the $i^{th}$ stock in the dataset. The $\sigma$-algebra $\mathcal F_t$ is generated by $\bigcup_i\{(S_u^i,V_u^i)_{u=0,1,\dots,t}\}$. We train the neural network separately to optimally execute for each individual stock, but the inputs to the network include the vectors of all prices and volumes at time $t$, which we denote as $\mathbf{S}_t=(S_t^1,S_t^2,\dots,S_t^{100})$ and $\mathbf{V}_t=(V_t^1,V_t^2,\dots,V_t^{100})$, respectively.

\subsection{Algorithm \& LSTM Training}

Because suborders are executed every five minutes, the total number of executions during the trading day is $390/5=78$. The output of the LSTM is $X_t^i$ which represents the remaining inventory for stock $i$ at time $t$. We train the LSTM network on each of the nine folds for each of the S\&P100 stocks. The total number of trained LSTM networks for all nine folds is $100\times 9=900$. The loss function used to train the LSTM network for stock $i$ is the following empirical approximation of  \eqref{eq:optimal_exec}, 
\begin{equation}
L^i=C_{\epsilon, \beta}\sum_{\ell=1}^{78}S_{5\ell}^i(V_{5\ell}^i)^{-\frac{1}{\beta+1}}|a_{5\ell-1}^i|^{\frac{\beta+2}{\beta+1}} = C_{\epsilon, \beta}\sum_{\ell=1}^{78}S_{5\ell}^i(V_{5\ell}^i)^{-\frac{1}{\beta+1}}|X_{5\ell}^i - X_{5\ell - 5}^i|^{\frac{\beta+2}{\beta+1}}  \ . 
\label{eq:transactioncosts}
\end{equation}
Define the LSTM network weights as $w$, then the training problem becomes
\begin{align*}
    w^* = \argmin_w L^i(w) = \argmin_w C_{\epsilon, \beta}\sum_{\ell=1}^{78}S_{5\ell}^i(V_{5\ell}^i)^{-\frac{1}{\beta+1}}|X_{5\ell}^i(w) - X_{5\ell - 5}^i(w)|^{\frac{\beta+2}{\beta+1}}.
\end{align*}

\begin{algorithm} 
\captionsetup{width=.9\linewidth}
 \caption{Training of LSTM Architecture for Optimal Execution Every 5 Minutes In A 390-Minute Trading Day for Stock $i$} 
 \label{alg:LSTMtrain}
 \begin{algorithmic} 
     \STATE --Initialize parameters of LSTM units $w$
     \FOR{k = 1 to NUM\_EPOCH}
        \STATE $X_0^i=x_0^i$, $L^i=0$, $t=1$, $\beta = 0.67$, $lr= 0.001$, $h_0^i = $ None
     	\WHILE{$t<390$}
        \STATE $X_t^i,h_t^i = LSTM(h_{t-1}^i, t, \mathbf{S}_{t},\mathbf{V}_{t}, \mathbf{X}_t^V, \mathbf{X}_t^T)$
     	    \IF{mod(t,5)=0}
     	        \STATE $L^i\mathrel{+}= C_{\epsilon, \beta} S_t^i(V_t^i)^{-\tfrac{1}{\beta+1}}|X_t^i-X_{t-5}^i|^{\tfrac{\beta+2}{\beta+1}} $
     	    \ENDIF
             \STATE $t += 1$
	    \ENDWHILE
     \STATE \#\#\#\# Close all the positions \#\#\#\#
     \STATE $X_{390}^i=0$
	    \STATE $L^i\mathrel{+}= C_{\epsilon, \beta} S_{390}^i(V_{390}^i)^{-\tfrac{1}{\beta+1}}|X_{390}^i-X_{385}^i|^{\tfrac{\beta+2}{\beta+1}}$\\
	    \#\#\#\# Update the LSTM weights \#\#\#\#
	    \STATE $w = $ Adam($L^i$, $lr$, $w$)
     \ENDFOR
 \end{algorithmic}
 \end{algorithm}

Algorithm \ref{alg:LSTMtrain} shows the training procedure for stock $i$'s LSTM network. For each training we loop for 10,000 epochs, which allows us to train a single LSTM in less than 20 minutes using GPUs; using the CPU, the runtime is around 10 hours. We initialize the LSTM state $h^i$ as {\bf None}. We also initialize $h^i$ with random numbers but  observe little difference. Adam is used as the optimizer with a learning rate ($lr$) of $0.001$. The initial inventory is 5\% of the stock's average daily volume, that is $x_0^i= .05\times A^i$ where 
\[A^i =  \hbox{sample mean of daily volume for stock $i$.}\] 
We also train for a fixed amount of shares for each stock, i.e., $x_0^i\equiv 10^6$. Note that although the time increments are 5 minutes, the data in between trade times is still seen by the LSTM, which means the strategy makes full use of the available information. The trading day has exactly 390 minutes, and so the first trade occurs at time $t=5$, and the final trade occurs at minute 390 when the market closes. The LSTM networks are trained on nine folds, each comprising 60 days of one-minute data, as described in Table~\ref{tab:dataset} for the S\&P100 stocks. We consider the days to be independent of each other, therefore, the shape of the LSTM training input is (60, 390, 401). The ratio of training data size ($60\times 390\times 401 = 9,383,400$) to trained parameter size (116,100) is over 80. Therefore, over-fitting is highly unlikely to occur. 

\begin{figure}
     \centering
     \includegraphics[width=\textwidth]{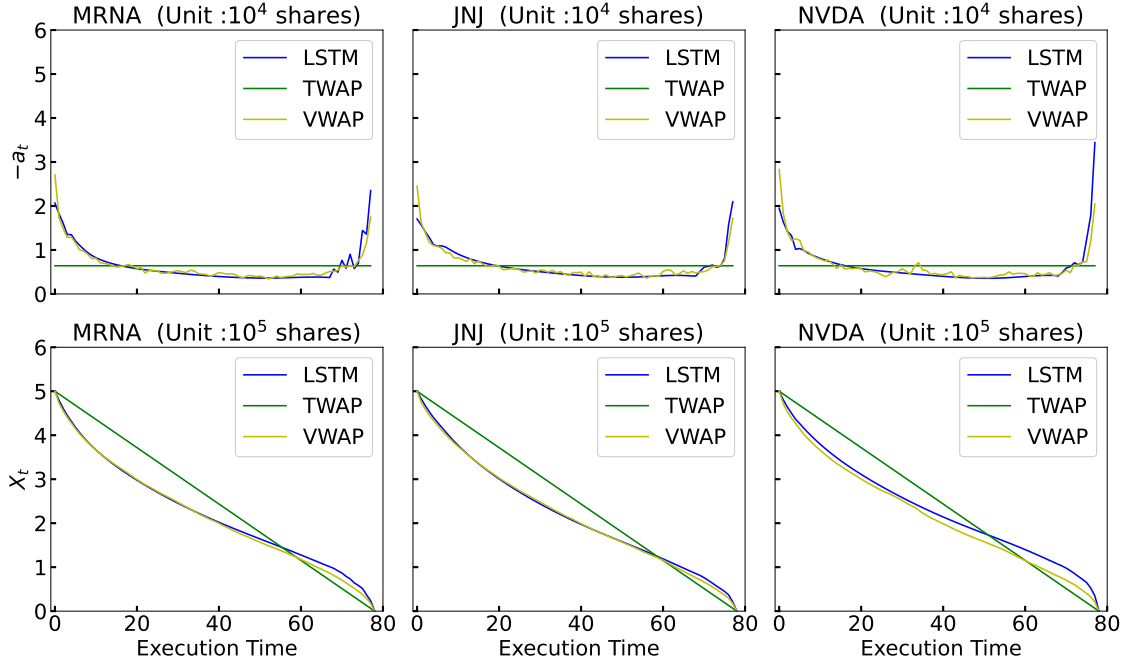}
     \caption{The first row shows $-a_t$ and the second row shows $X_t$ for selected stocks.}
     \label{fig:execution}
 \end{figure}

 The $\mathbf{X}_t^V$ and $\mathbf{X}_t^T$ in Alg.~\ref{alg:LSTMtrain} represent the remaining inventories for the S\&P100 stocks by using VWAP and TWAP strategies at the minute $t$, i.e.,
 \begin{align}
        \nonumber
      \mathbf{X}_t^T &= \left(x_0^1 - \frac{t}{T}x_0^1,\dots,x_0^{100} - \frac{t}{T}x_0^{100}\right) \\
      \label{eq:benchmark}
      \mathbf{X}_t^V &= \left(x_0^1 - \frac{ \sum_{j=1}^t\overline V_j^1}{\sum_{t=1}^T\overline V_t^1}x_0^1,\dots, x_0^{100} - \frac{ \sum_{j=1}^t\overline V_j^{100}}{\sum_{t=1}^T\overline V_t^{100}}x_0^{100}\right)\ .
 \end{align} 
 The TWAP strategy does not require the estimation of any parameters. As for the VWAP strategy, we utilized the 60 days' training volume data to estimate $\overline V_t^i$ for each fold and stock $i$ as described in \eqref{eq:VWAPstrat}, which is $\overline{V_t^i}=\left(\mathbb E(V_t^i)^{-\frac{1}{\beta+1}}\right)^{-(\beta+1)}$.  As an example, Fig.~\ref{fig:execution} shows the evolution of the actions and remaining inventories for the three strategies for three different stocks. During testing, the LSTM strategy is compared with the same TWAP and VWAP strategies in \eqref{eq:benchmark}.

\section{LSTM Execution Policy -- Experimental Results}
\subsection{Evaluation Metrics}
The metric used to compare our LSTM strategy with VWAP and TWAP strategies is the transaction cost $L^i$  as in \eqref{eq:transactioncosts}.
A smaller $L^i$  means a better performance of the selected execution strategy for stock $i$. 
We consider the following scenarios: the noiseless order book case with $x_0^i=.05\times A^i$ and the fixed amount of shares $x_0^i=10^6$, and the noisy order book case (i.e., $\beta$ is stochastic) with $x_0^i=.05\times A^i$ and $x_0^i=10^6$. \textbf{From here onward, we drop the super-script $i$ for notational simplicity except in places where it is necessary to show our results.}

\subsection{Noiseless Order Book Case}

We test the trained LSTM models on the datasets shown in Table~\ref{tab:dataset}. 
The parameter $\epsilon$ in \eqref{eq:liquid_ratio} is set to $0.006$ for $x_0= .05\times A$ and is set to $0.003$ for $x_0=10^6$ so that the transaction cost is approximately 1-2bps of the traded equity values.  
The first row in Fig.~\ref{fig:noiseless_foldcost} shows the average daily transaction cost for each fold. On average over the nine folds, the daily transaction cost  is \$811,433  for the LSTM strategy, \$843,248 for the VWAP strategy, and \$939,695  for the TWAP strategy for the case $x_0= .05\times A$. For the fixed initial shares case, the daily transaction cost is \$5,748,601 for the LSTM,  \$6,135,958 for the VWAP, and \$6,760,903 for the TWAP. Fig.~\ref{fig:noiseless_foldloss} shows how the daily transaction cost evolves with execution time for the case $x_0= .05\times A$. The LSTM strategy tends to incur large transaction costs towards the end of the trading session, whereas the VWAP strategy tends to incur large transaction costs in the morning and the TWAP strategy incurs  transaction costs relatively consistently throughout the day. However, the LSTM has an overall lower transaction cost. Similar execution behavior for the three strategies is observed in the fixed initial shares case.

\begin{figure}
    \centering
    \includegraphics[width=\textwidth]{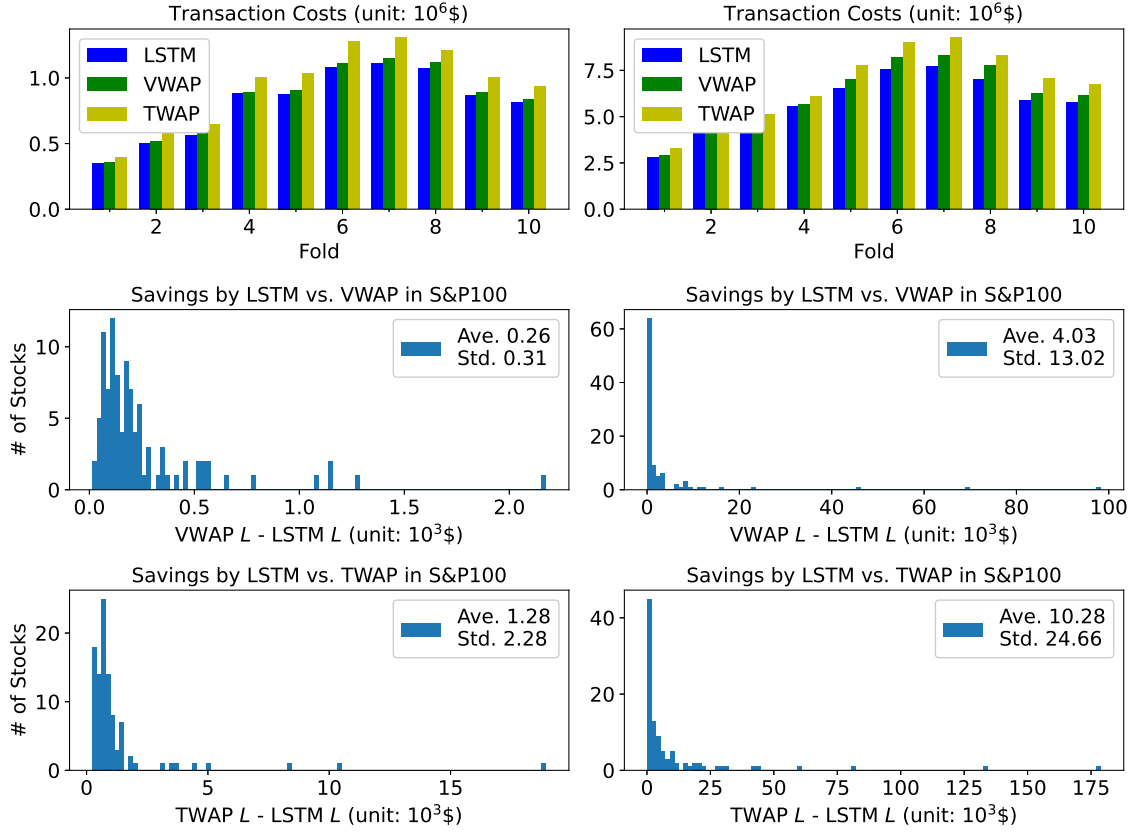}
    \caption{Performance of LSTM in noiseless order book case. Left: $x_0=0.05A$. Right: $x_0=10^6$.}
    \label{fig:noiseless_foldcost}
\end{figure}

\begin{figure}
    \centering
    \includegraphics[scale=0.3]{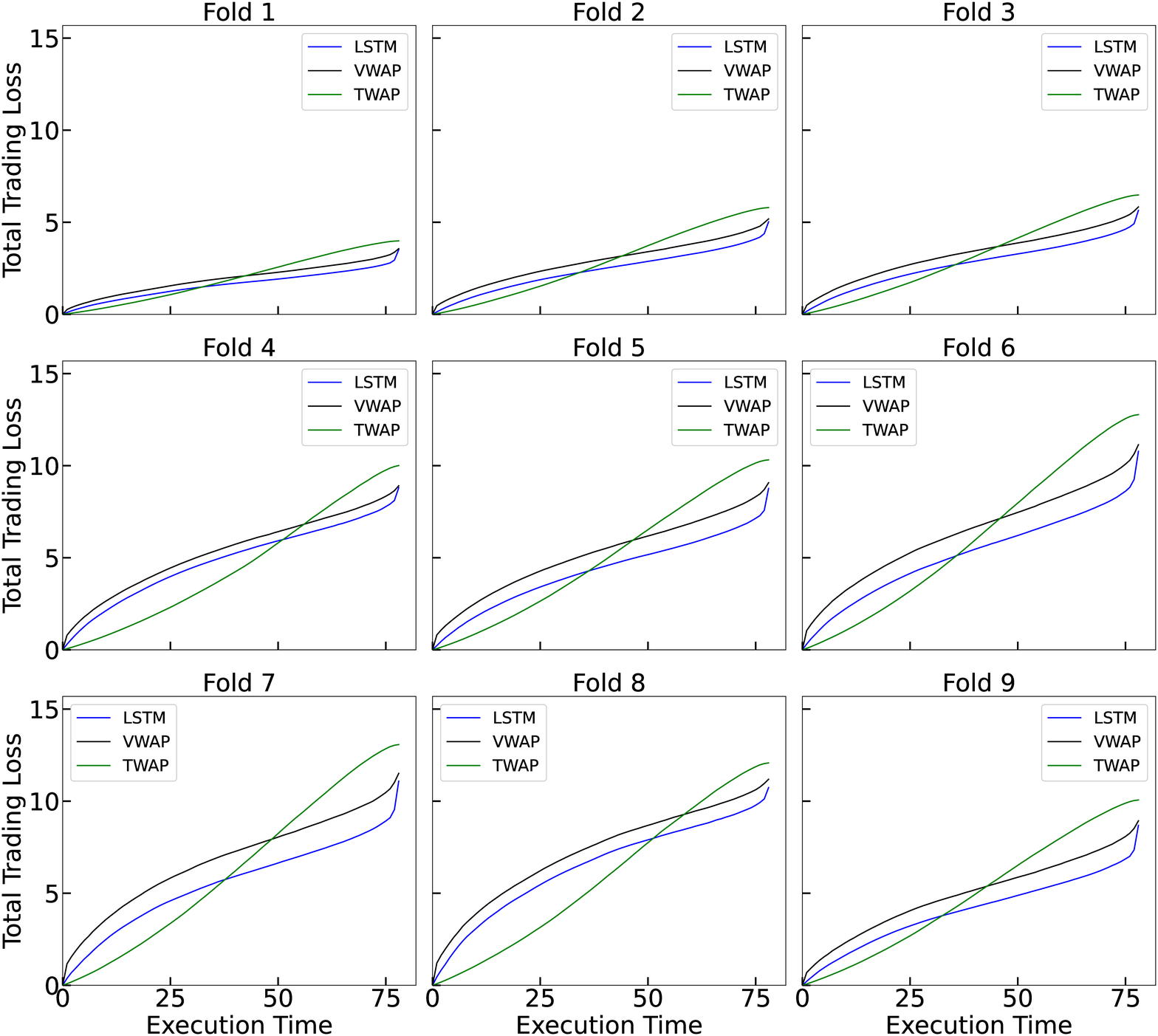}
    \caption{Transaction cost with execution time for noiseless order book case with $x_0=.05\times A$. The unit for $y$-axis is $10^5$ dollars.}
    \label{fig:noiseless_foldloss}
\end{figure}

The second row in Fig.~\ref{fig:noiseless_foldcost} shows the transaction costs saved by using the LSTM strategy compared to the VWAP strategy for each stock over the nine testing folds. The LSTM strategy outperforms the VWAP by having a smaller transaction cost for all the stocks in $x_0=.05\times A$ case and the average saving by LSTM is \$260 for each stock. As for the fixed initial shares case, the average savings become \$4,030. The third row shows the transaction costs saved by using the LSTM strategy compared to the TWAP strategy for each stock over the nine testing folds. Similarly, the LSTM is able to save \$1,280 in the percentage of daily volume case and \$10,280 in the fixed initial shares case. 

Tables~\ref{tab:noiseless_VWAP} and \ref{tab:noiseless_TWAP} list the top ten stock cases for which the LSTM saves the most and the bottom for which LSTM saves the least compared to VWAP and TWAP strategies, respectively, for the $x_0=.05\times A$ case. Most of the stocks in which the LSTM saves most are large capitalization tech companies. The median savings are approximately \$170  if LSTM is used rather than the VWAP strategy, and \$710 if LSTM is used rather than the TWAP strategy. Tables~\ref{tab:noiseless_VWAP_fixed} and \ref{tab:noiseless_TWAP_fixed} list the top and bottom ten stocks in which the LSTM saves the most/least compared to VWAP and TWAP strategies, respectively, for the $x_0=10^6$ case. A median saving of approximately \$550 is achieved by using LSTM rather than VWAP strategy, and \$2,360 by using LSTM rather than TWAP strategy.

\begin{table}
    \centering
    \caption{The ten stocks with highest and the ten with lowest savings under LSTM compared with VWAP strategy for noiseless order book case with $x_0=.05\times A$ (5\% of average volume), and the stock with the median saving.}
    \begin{tabular}{|c|ccccc|}
    \hline
      Ticker & Daily &  Equity Value  & VWAP (bps)  & LSTM (bps) & Savings (bps)  \\
     $\downarrow$ & Volume & of Traded (\$) &  (unit: \$$10^3$) & (unit: \$$10^3$) &   (unit: \$$10^3$) \\
         \hline
       AMZN & 42.65 M & 329.26 M & 55.98 (1.70) & 53.81 (1.63) & 2.18 (0.07) \\
TSLA & 28.20 M & 909.89 M & 162.77 (1.79) & 161.49 (1.77) & 1.29 (0.01) \\
MSFT & 20.61 M & 256.82 M & 35.02 (1.36) & 33.87 (1.32) & 1.16 (0.04) \\
AAPL & 87.88 M & 570.61 M & 77.39 (1.36) & 76.24 (1.34) & 1.14 (0.02) \\
GOOG & 0.62 M & 66.42 M & 12.20 (1.84) & 11.11 (1.67) & 1.08 (0.16) \\
NFLX & 3.85 M & 93.59 M & 18.52 (1.98) & 17.73 (1.89) & 0.78 (0.08) \\
UNH & 2.07 M & 39.53 M & 6.82 (1.73) & 6.17 (1.56) & 0.65 (0.16) \\
TMO & 0.93 M & 22.87 M & 4.33 (1.89) & 3.75 (1.64) & 0.57 (0.25) \\
NVDA & 32.57 M & 274.19 M & 39.69 (1.45) & 39.12 (1.43) & 0.56 (0.02) \\
META & 15.87 M & 219.53 M & 31.12 (1.42) & 30.59 (1.39) & 0.54 (0.02) \\
\hline
 $\vdots$ & \multicolumn{5}{c|}{$\cdots$} \\
 PM & 3.24 M & 13.50 M & 2.09 (1.55) & 1.93 (1.43) & 0.17 (0.12) \\
 $\vdots$ & \multicolumn{5}{c|}{$\cdots$} \\
       \hline
       AIG & 4.07 M & 8.88 M & 1.43 (1.61) & 1.37 (1.54) & 0.06 (0.07) \\
DUK & 1.93 M & 8.90 M & 1.33 (1.50) & 1.27 (1.43) & 0.06 (0.07)
 \\
 COP & 6.52 M & 18.26 M & 2.61 (1.43) & 2.55 (1.40) & 0.06 (0.03) \\
SO & 3.12 M & 9.28 M & 1.32 (1.42) & 1.27 (1.37) & 0.05 (0.06) \\
CL & 2.94 M & 11.28 M & 1.61 (1.43) & 1.56 (1.38) & 0.05 (0.05) \\
EXC & 5.72 M & 9.33 M & 1.35 (1.45) & 1.30 (1.40) & 0.04 (0.05) \\
WBA & 4.62 M & 10.07 M & 1.49 (1.48) & 1.45 (1.44) & 0.04 (0.04) \\
PG & 4.69 M & 31.84 M & 4.29 (1.35) & 4.25 (1.34) & 0.04 (0.01) \\
DD & 3.24 M & 10.72 M & 1.90 (1.77) & 1.87 (1.75) & 0.03 (0.02) \\
BK & 3.88 M & 8.69 M & 1.29 (1.48) & 1.27 (1.46) & 0.01 (0.02) \\
       \hline
    \end{tabular}
    \label{tab:noiseless_VWAP}
\end{table}

\begin{table}
    \centering
    \caption{The ten stocks with highest and the ten with lowest savings under LSTM compared with TWAP strategy for noiseless order book case with $x_0=.05\times A$ (5\% of average volume), and the stock with the median saving.}
    \begin{tabular}{|c|ccccc|}
    \hline
      Ticker & Daily & Equity Value  &  TWAP (bps)  & LSTM (bps) & Savings (bps)  \\
     $\downarrow$ & Volume & Traded (\$) &  (unit: \$$10^3$) & (unit: \$$10^3$) &   (unit: \$$10^3$) \\
         \hline
       TSLA & 28.20 M & 909.89 M & 180.43 (1.98) & 161.49 (1.77) & 18.94 (0.21) \\
AMZN & 42.65 M & 329.26 M & 64.19 (1.95) & 53.81 (1.63) & 10.38 (0.32) \\
AAPL & 87.88 M & 570.61 M & 84.61 (1.48) & 76.24 (1.34) & 8.37 (0.15) \\
NVDA & 32.57 M & 274.19 M & 44.06 (1.61) & 39.12 (1.43) & 4.94 (0.18) \\
MSFT & 20.61 M & 256.82 M & 38.40 (1.50) & 33.87 (1.32) & 4.53 (0.18) \\
NFLX & 3.85 M & 93.59 M & 21.47 (2.29) & 17.73 (1.89) & 3.73 (0.40) \\
META & 15.87 M & 219.53 M & 34.09 (1.55) & 30.59 (1.39) & 3.50 (0.16) \\
BA & 14.42 M & 139.82 M & 26.20 (1.87) & 23.01 (1.65) & 3.19 (0.23) \\
GOOG & 0.62 M & 66.42 M & 13.14 (1.98) & 11.11 (1.67) & 2.02 (0.30) \\
PYPL & 7.05 M & 70.38 M & 13.17 (1.87) & 11.34 (1.61) & 1.83 (0.26) \\
\hline
 $\vdots$ & \multicolumn{5}{c|}{$\cdots$} \\
 ORCL & 8.61 M & 30.58 M & 5.21 (1.70) & 4.50 (1.47) & 0.71 (0.23) \\
  $\vdots$ & \multicolumn{5}{c|}{$\cdots$} \\
       \hline
       AIG & 4.07 M & 8.88 M & 1.71 (1.92) & 1.37 (1.54) & 0.34 (0.38) \\
EMR & 2.12 M & 8.54 M & 1.61 (1.89) & 1.28 (1.50) & 0.33 (0.39) \\
MET & 3.72 M & 9.63 M & 1.70 (1.76) & 1.38 (1.44) & 0.31 (0.32) \\
DUK & 1.93 M & 8.90 M & 1.58 (1.78) & 1.27 (1.43) & 0.31 (0.35) \\
SO & 3.12 M & 9.28 M & 1.57 (1.70) & 1.27 (1.37) & 0.31 (0.33) \\
CL & 2.94 M & 11.28 M & 1.86 (1.65) & 1.56 (1.38) & 0.29 (0.26) \\
DOW & 3.72 M & 9.69 M & 1.62 (1.67) & 1.36 (1.41) & 0.25 (0.26) \\
BKNG & 0.22 M & 22.54 M & 3.96 (1.75) & 3.70 (1.64) & 0.25 (0.11) \\
WBA & 4.62 M & 10.07 M & 1.70 (1.68) & 1.45 (1.44) & 0.25 (0.25) \\
BK & 3.88 M & 8.69 M & 1.51 (1.73) & 1.27 (1.46) & 0.23 (0.27) \\
       \hline
    \end{tabular}
    \label{tab:noiseless_TWAP}
\end{table}

\begin{table}
    \centering
    \caption{The ten stocks with highest and the ten with lowest savings under LSTM compared with VWAP strategy for noiseless order book case with $x_0=10^6$, and the stock with the median saving.}
    \begin{tabular}{|c|ccccc|}
    \hline
      Ticker & Daily &  Equity Value  & VWAP (bps)  & LSTM (bps) & Savings (bps)  \\
     $\downarrow$ & Volume & of Traded (\$) &  (unit: \$$10^3$) & (unit: \$$10^3$) &   (unit: \$$10^3$) \\
         \hline
       GOOG & 0.62 M & 2150.54 M & 1046.09 (4.86) & 947.77 (4.41) & 98.32 (0.46) \\
BKNG & 0.22 M & 2076.48 M & 1741.56 (8.39) & 1671.94 (8.05) & 69.62 (0.34) \\
BLK & 0.46 M & 707.64 M & 439.79 (6.21) & 394.23 (5.57) & 45.57 (0.64) \\
TMO & 2.98 M & 493.96 M & 194.35 (3.93) & 171.52 (3.47) & 22.82 (0.46) \\
AVGO & 1.08 M & 435.77 M & 149.75 (3.44) & 133.44 (3.06) & 16.32 (0.37) \\
CHTR & 0.61 M & 616.99 M & 289.98 (4.70) & 278.15 (4.51) & 11.83 (0.19) \\
ADBE & 1.43 M & 494.80 M & 140.58 (2.84) & 129.16 (2.61) & 11.42 (0.23) \\
LMT & 0.96 M & 365.99 M & 126.62 (3.46) & 117.34 (3.21) & 9.28 (0.25) \\
UNH & 2.89 M & 381.23 M & 84.43 (2.21) & 75.98 (1.99) & 8.45 (0.22) \\
ACN & 1.17 M & 276.23 M & 81.30 (2.94) & 72.95 (2.64) & 8.35 (0.30) \\
\hline
 $\vdots$ & \multicolumn{5}{c|}{$\cdots$} \\
 ABT & 4.05 M & 111.71 M & 14.13 (1.27) & 13.58 (1.22) & 0.55 (0.05) \\
 $\vdots$ & \multicolumn{5}{c|}{$\cdots$} \\
       \hline
       BA & 14.42 M & 193.90 M & 12.90 (0.67) & 12.82 (0.66) & 0.08 (0.00) \\
BK & 3.88 M & 44.78 M & 5.83 (1.30) & 5.75 (1.28) & 0.08 (0.02) \\
GM & 13.78 M & 44.10 M & 2.48 (0.56) & 2.43 (0.55) & 0.05 (0.01) \\
C & 18.34 M & 56.97 M & 2.74 (0.48) & 2.69 (0.47) & 0.05 (0.01) \\
AAPL & 87.88 M & 129.86 M & 2.40 (0.18) & 2.36 (0.18) & 0.04 (0.00) \\
WFC & 14.04 M & 37.81 M & 1.46 (0.39) & 1.43 (0.38) & 0.03 (0.01) \\
XOM & 27.13 M & 53.12 M & 2.25 (0.42) & 2.23 (0.42) & 0.03 (0.01) \\
PFE & 25.04 M & 40.06 M & 1.59 (0.40) & 1.57 (0.39) & 0.02 (0.01) \\
BAC & 47.20 M & 34.29 M & 0.95 (0.28) & 0.93 (0.27) & 0.02 (0.01) \\
F & 70.81 M & 11.81 M & 0.26 (0.22) & 0.26 (0.22) & 0.00 (0.00) \\
       \hline
    \end{tabular}
    \label{tab:noiseless_VWAP_fixed}
\end{table}

\begin{table}
    \centering
     \caption{The ten stocks with highest and the ten with lowest savings under LSTM compared with TWAP strategy for noiseless order book case with $x_0=10^6$, and the stock with the median saving.}
    \begin{tabular}{|c|ccccc|}
    \hline
      Ticker & Daily & Equity Value  &  TWAP (bps)  & LSTM (bps) & Savings (bps)  \\
     $\downarrow$ & Volume & Traded (\$) &  (unit: \$$10^3$) & (unit: \$$10^3$) &   (unit: \$$10^3$) \\
         \hline
       GOOG & 0.62 M & 2150.54 M & 1126.70 (5.24) & 947.77 (4.41) & 178.93 (0.83) \\
BKNG & 0.22 M & 2076.48 M & 1805.40 (8.69) & 1671.94 (8.05) & 133.46 (0.64) \\
BLK & 0.46 M & 707.64 M & 476.08 (6.73) & 394.23 (5.57) & 81.85 (1.16) \\
CHTR & 0.61 M & 616.99 M & 337.32 (5.47) & 278.15 (4.51) & 59.17 (0.96) \\
AVGO & 1.08 M & 435.77 M & 177.81 (4.08) & 133.44 (3.06) & 44.37 (1.02) \\
TMO & 0.93 M & 493.96 M & 214.06 (4.33) & 171.52 (3.47) & 42.53 (0.86) \\
ADBE & 1.43 M & 494.80 M & 161.14 (3.26) & 129.16 (2.61) & 31.97 (0.65) \\
LIN & 1.07 M & 270.08 M & 109.40 (4.05) & 79.80 (2.95) & 29.59 (1.10) \\
LMT & 0.96 M & 365.99 M & 144.76 (3.96) & 117.34 (3.21) & 27.41 (0.75) \\
COST & 1.45 M & 404.42 M & 124.43 (3.08) & 102.18 (2.53) & 22.24 (0.55) \\
\hline
 $\vdots$ & \multicolumn{5}{c|}{$\cdots$} \\
 JNJ & 4.49 M & 156.07 M & 19.59 (1.26) & 17.23 (1.10) & 2.36 (0.15) \\
  $\vdots$ & \multicolumn{5}{c|}{$\cdots$} \\
       \hline
       C & 18.34 M & 56.97 M & 3.08 (0.54) & 2.69 (0.47) & 0.39 (0.07) \\
INTC & 22.63 M & 51.50 M & 2.50 (0.48) & 2.11 (0.41) & 0.38 (0.07) \\
GM & 13.78 M & 44.10 M & 2.78 (0.63) & 2.43 (0.55) & 0.35 (0.08) \\
XOM & 5.87 M & 53.12 M & 2.52 (0.47) & 2.23 (0.42) & 0.29 (0.05) \\
AAPL & 87.88 M & 129.86 M & 2.62 (0.20) & 2.36 (0.18) & 0.26 (0.02) \\
PFE & 25.04 M & 40.06 M & 1.82 (0.46) & 1.57 (0.39) & 0.25 (0.06) \\
WFC & 4.62 M & 37.81 M & 1.65 (0.44) & 1.43 (0.38) & 0.22 (0.06) \\
BAC & 47.20 M & 34.29 M & 1.08 (0.32) & 0.93 (0.27) & 0.15 (0.04) \\
T & 45.78 M & 19.94 M & 0.61 (0.30) & 0.55 (0.27) & 0.06 (0.03) \\
F & 70.81 M & 11.81 M & 0.30 (0.26) & 0.26 (0.22) & 0.04 (0.03)  \\
       \hline
    \end{tabular}
    \label{tab:noiseless_TWAP_fixed}
\end{table}

\subsection{Noisy Order Book Shape}
\label{subsec:noisy}

The $\beta=.67$ estimated in \cite{almgren2005direct} is an average. In any given minute the limit order books may have a different power law near to but not equal to $.67$. To model this variation, we consider $\beta$ to be stochastic, i.e.,
\begin{align}
    \beta_{t} = 0.67 + \eta_{t},
\end{align} where $\eta_{t}$ is a random variable with uniform distribution on $(-0.3, 0.3)$. The range of noise is $0.6 / 0.67 \approx 0.9$. The trading loss then becomes 
\begin{align}
    \label{eq:noisy_loss}
   L =  \sum_{\ell=1}^{78} C_{\epsilon, \beta_{5\ell}} S_{5\ell}(V_{5\ell})^{-\tfrac{1}{\beta_{5\ell}+1}}|X_{5\ell} - X_{5\ell-5}|^{\tfrac{\beta_{5\ell}+2}{\beta_{5\ell}+1}},
\end{align} where $C_{\epsilon, \beta_{5\ell}}$ is the stochastic version of $C_{\epsilon, \beta}$ obtained by substituting $\beta_{5\ell}$ into \eqref{eq:liquid_ratio}.

\begin{figure}
    \centering
    \includegraphics[width=\textwidth]{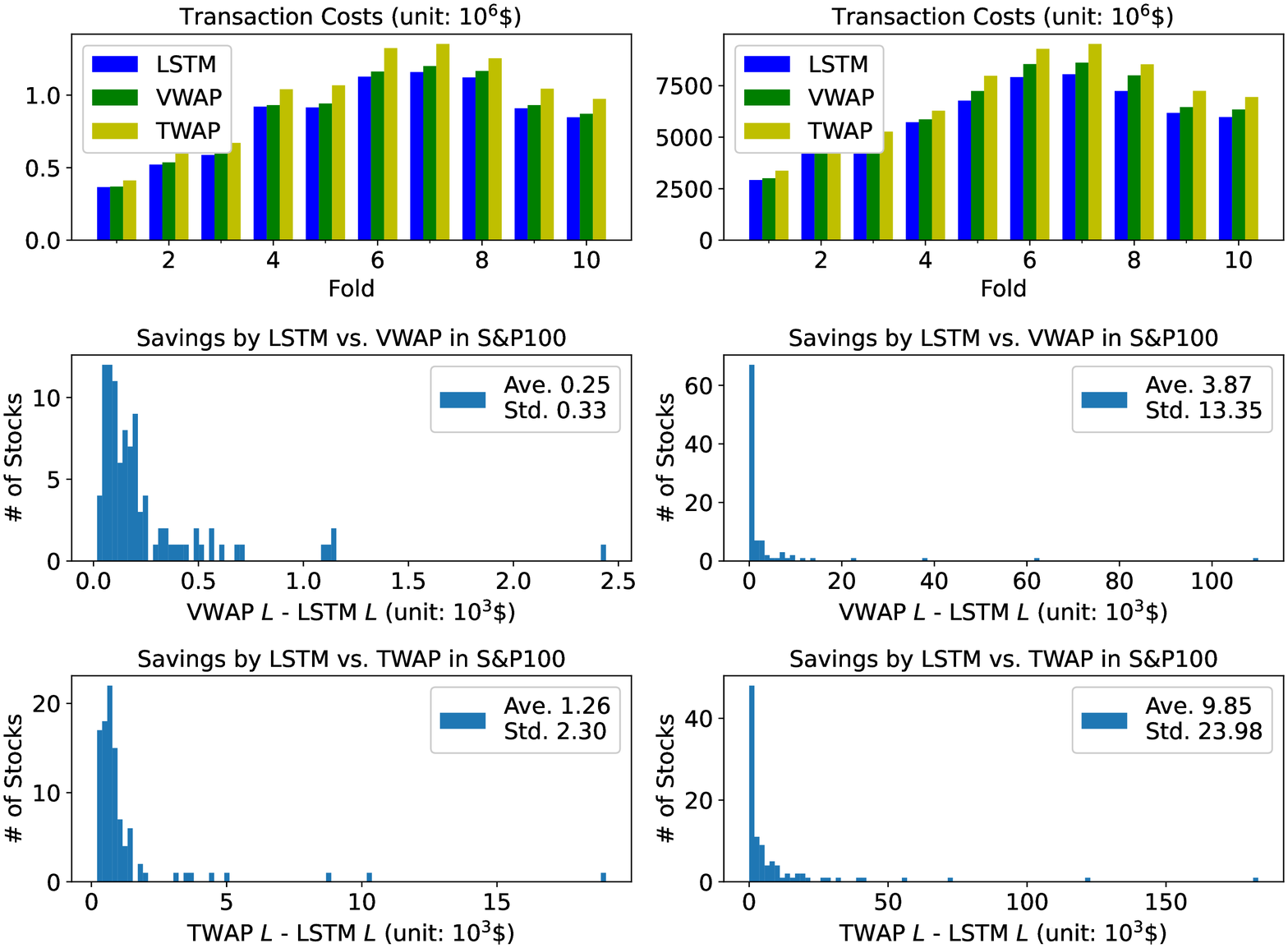}
    \caption{Performance of LSTM in noisy order book case. Left: $x_0=0.05A$. Right: $x_0=10^6$.}
    \label{fig:noisy_costs}
\end{figure}

The first row in Fig.~\ref{fig:noisy_costs} shows the average daily transaction cost $L$ for each fold in the noisy order book case. On average over the nine folds, the daily transaction cost  is \$847,328  for the LSTM strategy, \$871,606 for the VWAP strategy, and \$973,718  for the TWAP strategy for $x_0=.05\times A$. As for $x_0=10^6$, the daily transaction cost is \$5,975,930 for the LSTM, \$6,344,318 for the VWAP, and \$6,946,579 for the TWAP. Compared with the noiseless order book shape case, the transaction costs increase for all the strategies. However, the LSTM strategy still consistently has a smaller transaction cost than the VWAP and TWAP strategies. Note that the LSTM networks used are the same as in the noiseless order book shape case. Therefore, the LSTM strategy is robust to noise perturbations in $\beta$. As expected, we empirically observe that the difference between LSTM and VWAP in daily transaction cost decreases with the increase in noise intensity. The second row in Fig.~\ref{fig:noisy_costs} shows the transaction costs saved by using the LSTM strategy rather than the VWAP and TWAP strategies for each stock. Compared with the noiseless order book case, the average savings becomes smaller.  Similarly,  Tables~\ref{tab:noisy_VWAP} and \ref{tab:noisy_TWAP} show the top/bottom ten stock cases in which the LSTM saves the most/least transaction costs compared to VWAP and TWAP strategies, respectively, for $x_0=.05\times A$. The median savings are around \$160 for the VWAP case and \$690 for the TWAP case. Tables~\ref{tab:noisy_VWAP_fixed} and \ref{tab:noisy_TWAP_fixed} show the top/bottom ten stock cases for $x_0=10^6$. The median savings are around \$400 for the VWAP case and \$2,070 for the TWAP case. 

\begin{table}
    \centering
    \caption{The ten stocks with highest and the ten with lowest savings under LSTM compared with VWAP strategy for noisy order book case with $x_0=.05\times A$ (5\% of average volume), and the stock with the median saving.}
    \begin{tabular}{|c|ccccc|}
    \hline
      Ticker & Daily &  Equity Value  & VWAP (bps)  & LSTM (bps) & Savings (bps)  \\
     $\downarrow$ & Volume & of Traded (\$) &  (unit: \$$10^3$) & (unit: \$$10^3$) &   (unit: \$$10^3$) \\
         \hline
       AMZN & 42.65 M & 329.26 M & 59.85 (1.82) & 57.41 (1.74) & 2.44 (0.07) \\
TSLA & 0.93 M & 909.89 M & 171.62 (1.89) & 170.47 (1.87) & 1.15 (0.01) \\
GOOG & 13.78 M & 66.42 M & 12.25 (1.84) & 11.11 (1.67) & 1.14 (0.17) \\
AAPL & 87.88 M & 570.61 M & 83.39 (1.46) & 82.29 (1.44) & 1.11 (0.02) \\
MSFT & 8.54 M & 256.82 M & 36.59 (1.42) & 35.48 (1.38) & 1.11 (0.04) \\
NFLX & 6.03 M & 93.59 M & 18.77 (2.00) & 18.05 (1.93) & 0.71 (0.08) \\
UNH & 28.20 M & 39.53 M & 6.92 (1.75) & 6.24 (1.58) & 0.68 (0.17) \\
NVDA & 4.60 M & 274.19 M & 42.02 (1.53) & 41.41 (1.51) & 0.62 (0.02) \\
META & 3.72 M & 219.53 M & 32.52 (1.48) & 31.96 (1.46) & 0.57 (0.03) \\
TMO & 45.78 M & 22.87 M & 4.36 (1.91) & 3.80 (1.66) & 0.56 (0.25) \\
\hline
 $\vdots$ & \multicolumn{5}{c|}{$\cdots$} \\
 GD & 1.88 M & 7.79 M & 1.45 (1.86) & 1.29 (1.66) & 0.16 (0.20) \\
 $\vdots$ & \multicolumn{5}{c|}{$\cdots$} \\
       \hline
       KHC & 10.87 M & 8.63 M & 1.46 (1.70) & 1.41 (1.63) & 0.06 (0.07) \\
TGT & 2.96 M & 28.09 M & 4.82 (1.71) & 4.76 (1.69) & 0.06 (0.02)
 \\ MS & 8.03 M & 31.26 M & 4.76 (1.52) & 4.71 (1.51) & 0.05 (0.02) \\
DIS & 1.66 M & 66.71 M & 10.72 (1.61) & 10.67 (1.60) & 0.05 (0.01) \\
COP & 2.31 M & 18.26 M & 2.67 (1.46) & 2.62 (1.43) & 0.05 (0.03) \\
AIG & 4.07 M & 8.88 M & 1.46 (1.64) & 1.40 (1.58) & 0.05 (0.06) \\
DUK & 3.72 M & 8.90 M & 1.35 (1.52) & 1.30 (1.46) & 0.05 (0.05) \\
MO & 1.69 M & 14.41 M & 2.24 (1.56) & 2.20 (1.52) & 0.05 (0.03) \\
WBA & 5.69 M & 10.07 M & 1.52 (1.51) & 1.48 (1.47) & 0.04 (0.04) \\
SO & 5.93 M & 9.28 M & 1.34 (1.44) & 1.30 (1.40) & 0.04 (0.04)  \\
       \hline
    \end{tabular}
    \label{tab:noisy_VWAP}
\end{table}

\begin{table}
    \centering
    \caption{The ten stocks with highest and the ten with lowest savings under LSTM compared with TWAP strategy for noisy order book case with $x_0=.05\times A$ (5\% of average volume), and the stock with the median saving.}
    \begin{tabular}{|c|ccccc|}
    \hline
      Ticker & Daily & Equity Value  &  TWAP (bps)  & LSTM (bps) & Savings (bps)  \\
     $\downarrow$ & Volume & Traded (\$) &  (unit: \$$10^3$) & (unit: \$$10^3$) &   (unit: \$$10^3$) \\
         \hline
       TSLA & 28.20 M & 909.89 M & 189.50 (2.08) & 170.47 (1.87) & 19.03 (0.21) \\
AMZN & 42.65 M & 329.26 M & 67.78 (2.06) & 57.41 (1.74) & 10.37 (0.31) \\
AAPL & 87.88 M & 570.61 M & 91.02 (1.60) & 82.29 (1.44) & 8.73 (0.15) \\
NVDA & 32.57 M & 274.19 M & 46.33 (1.69) & 41.41 (1.51) & 4.92 (0.18) \\
MSFT & 20.61 M & 256.82 M & 40.02 (1.56) & 35.48 (1.38) & 4.53 (0.18) \\
NFLX & 3.85 M & 93.59 M & 21.74 (2.32) & 18.05 (1.93) & 3.69 (0.39) \\
META & 15.87 M & 219.53 M & 35.39 (1.61) & 31.96 (1.46) & 3.43 (0.16) \\
BA & 14.42 M & 139.82 M & 27.06 (1.94) & 23.84 (1.70) & 3.22 (0.23) \\
GOOG & 0.62 M & 66.42 M & 13.09 (1.97) & 11.11 (1.67) & 1.97 (0.30) \\
BAC & 47.20 M & 80.93 M & 13.73 (1.70) & 11.88 (1.47) & 1.85 (0.23) \\
\hline
 $\vdots$ & \multicolumn{5}{c|}{$\cdots$} \\
 FDX & 1.88 M & 21.55 M & 4.20 (1.95) & 3.51 (1.63) & 0.69 (0.32) \\
  $\vdots$ & \multicolumn{5}{c|}{$\cdots$} \\
       \hline
       COP & 6.52 M & 18.26 M & 2.95 (1.62) & 2.62 (1.43) & 0.33 (0.18) \\
KHC & 5.07 M & 8.63 M & 1.74 (2.01) & 1.41 (1.63) & 0.33 (0.39)  \\
AIG & 4.07 M & 8.88 M & 1.73 (1.95) & 1.40 (1.58) & 0.32 (0.36)  \\
EMR & 2.12 M & 8.54 M & 1.62 (1.90) & 1.31 (1.53) & 0.32 (0.37)  \\
CL & 2.94 M & 11.28 M & 1.87 (1.66) & 1.57 (1.40) & 0.30 (0.27)  \\
MET & 3.72 M & 9.63 M & 1.72 (1.78) & 1.42 (1.47) & 0.30 (0.31)  \\
SO & 3.12 M & 9.28 M & 1.59 (1.72) & 1.30 (1.40) & 0.29 (0.31) \\
DUK & 1.93 M & 8.90 M & 1.59 (1.79) & 1.30 (1.46) & 0.29 (0.32)  \\
DOW & 3.72 M & 9.69 M & 1.64 (1.69) & 1.39 (1.43) & 0.25 (0.26) \\
WBA & 4.62 M & 10.07 M & 1.72 (1.71) & 1.48 (1.47) & 0.24 (0.24) \\
       \hline
    \end{tabular}
    \label{tab:noisy_TWAP}
\end{table}

\begin{table}
    \centering
     \caption{The ten stocks with highest and the ten with lowest savings under LSTM compared with VWAP strategy for noisy order book case with $x_0=10^6$, and the stock with the median saving.}
    \begin{tabular}{|c|ccccc|}
    \hline
      Ticker & Daily &  Equity Value  & VWAP (bps)  & LSTM (bps) & Savings (bps)  \\
     $\downarrow$ & Volume & of Traded (\$) &  (unit: \$$10^3$) & (unit: \$$10^3$) &   (unit: \$$10^3$) \\
         \hline
       GOOG & 0.62 M & 2150.54 M & 1084.43 (5.04) & 974.45 (4.53) & 109.98 (0.51) \\
BKNG & 0.22 M & 2076.48 M & 1795.88 (8.65) & 1734.05 (8.35) & 61.83 (0.30) \\
BLK & 0.46 M & 707.64 M & 454.80 (6.43) & 417.02 (5.89) & 37.78 (0.53) \\
TMO & 45.78 M & 493.96 M & 201.56 (4.08) & 179.50 (3.63) & 22.06 (0.45) \\
AVGO & 1.08 M & 435.77 M & 155.41 (3.57) & 141.31 (3.24) & 14.09 (0.32) \\
ADBE & 1.43 M & 494.80 M & 145.13 (2.93) & 133.28 (2.69) & 11.85 (0.24) \\
CHTR & 0.61 M & 616.99 M & 299.32 (4.85) & 290.13 (4.70) & 9.19 (0.15) \\
LMT & 0.96 M & 365.99 M & 130.78 (3.57) & 121.84 (3.33) & 8.94 (0.24) \\
UNH & 28.20 M & 381.23 M & 87.05 (2.28) & 79.21 (2.08) & 7.84 (0.21) \\
GD & 0.90 M & 173.48 M & 70.89 (4.09) & 63.23 (3.64) & 7.65 (0.44) \\
\hline
 $\vdots$ & \multicolumn{5}{c|}{$\cdots$} \\
 CVS & 5.27 M & 77.55 M & 8.99 (1.16) & 8.59 (1.11) & 0.40 (0.05) \\
 $\vdots$ & \multicolumn{5}{c|}{$\cdots$} \\
       \hline
       BA & 14.42 M & 193.90 M & 13.29 (0.69) & 13.22 (0.68) & 0.07 (0.00) \\
EXC & 5.72 M & 32.63 M & 3.41 (1.05) & 3.35 (1.03) & 0.06 (0.02)  \\
KHC & 5.07 M & 34.06 M & 4.41 (1.30) & 4.36 (1.28) & 0.06 (0.02)  \\
GM & 13.78 M & 44.10 M & 2.55 (0.58) & 2.50 (0.57) & 0.05 (0.01)  \\
AAPL & 87.88 M & 129.86 M & 2.47 (0.19) & 2.43 (0.19) & 0.04 (0.00) \\
C & 18.34 M & 56.97 M & 2.82 (0.50) & 2.79 (0.49) & 0.03 (0.01)  \\
BK & 3.88 M & 44.78 M & 6.02 (1.34) & 5.99 (1.34) & 0.03 (0.01)  \\
WFC & 5.69 M & 37.81 M & 1.51 (0.40) & 1.48 (0.39) & 0.03 (0.01) \\
XOM & 4.62 M & 53.12 M & 2.32 (0.44) & 2.29 (0.43) & 0.02 (0.00) \\
BAC & 47.20 M & 34.29 M & 0.99 (0.29) & 0.97 (0.28) & 0.02 (0.01) \\
       \hline
    \end{tabular}
    \label{tab:noisy_VWAP_fixed}
\end{table}

\begin{table}
    \centering
    \caption{The ten stocks with highest and the ten with lowest savings under LSTM compared with TWAP strategy for noisy order book case with $x_0=10^6$, and the stock with the median saving.}
    \begin{tabular}{|c|ccccc|}
    \hline
      Ticker & Daily & Equity Value  &  TWAP (bps)  & LSTM (bps) & Savings (bps)  \\
     $\downarrow$ & Volume & Traded (\$) &  (unit: \$$10^3$) & (unit: \$$10^3$) &   (unit: \$$10^3$) \\
         \hline
       GOOG & 0.62 M & 2150.54 M & 1157.61 (5.38) & 974.45 (4.53) & 183.16 (0.85) \\
BKNG & 0.22 M & 2076.48 M & 1856.14 (8.94) & 1734.05 (8.35) & 122.09 (0.59) \\
BLK & 0.46 M & 707.64 M & 488.79 (6.91) & 417.02 (5.89) & 71.77 (1.01) \\
CHTR & 0.61 M & 616.99 M & 346.58 (5.62) & 290.13 (4.70) & 56.45 (0.91) \\
AVGO & 1.08 M & 435.77 M & 182.79 (4.19) & 141.31 (3.24) & 41.47 (0.95) \\
TMO & 0.93 M & 493.96 M & 219.81 (4.45) & 179.50 (3.63) & 40.31 (0.82) \\
ADBE & 1.43 M & 494.80 M & 165.55 (3.35) & 133.28 (2.69) & 32.27 (0.65) \\
LIN & 1.07 M & 270.08 M & 112.32 (4.16) & 83.64 (3.10) & 28.68 (1.06) \\
LMT & 0.96 M & 365.99 M & 148.46 (4.06) & 121.84 (3.33) & 26.62 (0.73) \\
COST & 1.45 M & 404.42 M & 127.93 (3.16) & 107.02 (2.65) & 20.91 (0.52) \\
\hline
 $\vdots$ & \multicolumn{5}{c|}{$\cdots$} \\
 QCOM & 6.20 M & 130.65 M & 14.88 (1.14) & 12.82 (0.98) & 2.07 (0.16) \\
  $\vdots$ & \multicolumn{5}{c|}{$\cdots$} \\
       \hline
       KO & 12.37 M & 51.75 M & 3.78 (0.73) & 3.25 (0.63) & 0.53 (0.10) \\
VZ & 5.69 M & 52.80 M & 3.35 (0.63) & 2.93 (0.55) & 0.42 (0.08) \\
C & 18.34 M & 56.97 M & 3.16 (0.55) & 2.79 (0.49) & 0.37 (0.07) \\
INTC & 22.63 M & 51.50 M & 2.56 (0.50) & 2.19 (0.43) & 0.37 (0.07) \\
GM & 13.78 M & 44.10 M & 2.85 (0.65) & 2.50 (0.57) & 0.35 (0.08) \\
XOM & 5.87 M & 53.12 M & 2.58 (0.49) & 2.29 (0.43) & 0.29 (0.05) \\
AAPL & 87.88 M & 129.86 M & 2.69 (0.21) & 2.43 (0.19) & 0.26 (0.02) \\
PFE & 25.04 M & 40.06 M & 1.87 (0.47) & 1.63 (0.41) & 0.24 (0.06) \\
WFC & 4.62 M & 37.81 M & 1.70 (0.45) & 1.48 (0.39) & 0.22 (0.06) \\
BAC & 47.20 M & 34.29 M & 1.11 (0.32) & 0.97 (0.28) & 0.14 (0.04) \\
       \hline
    \end{tabular}
    \label{tab:noisy_TWAP_fixed}
\end{table}

\section{Conclusion}
We have shown how LSTM can be used for optimal execution of large stock orders in a limit order book. Our backtests demonstrate that LSTM can outperform TWAP and VWAP-based strategies in order book models with both noiseless and noisy power-law parameter. It is possible that the improved performance of the LSTM is due to its ability to aggregate across multiple stocks and exploit co-dependence in the price and volume time series. 
There are a variety of future avenues to continue this work. One such direction is to include permanent price impact and to see how LSTM can adjust to early suborders adversely affecting price. Another direction would be to consider optimizing length of trading period and frequency of trading, both of which were static hyperparameters in this paper.

\section*{Declaration of Interest}
This work was  supported in part by NSF grant DMS-1907518 and in part by the New York University Abu Dhabi (NYUAD) Center for Artificial Intelligence and Robotics, funded by Tamkeen under the NYUAD Research Institute Award CG010.

\bibliography{refs.bib}

\begin{thebibliography}{10}

\bibitem{alfonsi2016dynamic}
Aur{\'e}lien Alfonsi and Pierre Blanc.
\newblock Dynamic optimal execution in a mixed-market-impact {H}awkes price
  model.
\newblock {\em Finance and Stochastics}, 20(1):183--218, 2016.

\bibitem{qb20201008}
Robert Almgren.
\newblock Futures cost model.
\newblock Technical report, Quantitative Brokers, Oct 2020.

\bibitem{qb20200424}
Robert Almgren.
\newblock Transaction cost model in {M}arch 2020.
\newblock Technical report, Quantitative Brokers, Apr 2020.

\bibitem{almgren2001optimal}
Robert Almgren and Neil Chriss.
\newblock Optimal execution of portfolio transactions.
\newblock {\em Journal of Risk}, 3:5--40, 2001.

\bibitem{almgren2005direct}
Robert Almgren, Chee Thum, Emmanuel Hauptmann, and Hong Li.
\newblock Direct estimation of equity market impact.
\newblock {\em Risk}, 18(7):58--62, 2005.

\bibitem{amaral2019price}
Lucas Amaral and Andrew Papanicolaou.
\newblock Price impact of large orders using {H}awkes processes.
\newblock {\em ANZIAM}, 61(2):161 -- 194, 2019.

\bibitem{bacry2015hawkes}
Emmanuel Bacry, Iacopo Mastromatteo, and Jean-Fran{\c{c}}ois Muzy.
\newblock Hawkes processes in finance.
\newblock {\em Market Microstructure and Liquidity}, 1(01):1550005, 2015.

\bibitem{bershova2013non}
Nataliya Bershova and Dmitry Rakhlin.
\newblock The non-linear market impact of large trades: Evidence from buy-side
  order flow.
\newblock {\em Quantitative finance}, 13(11):1759--1778, 2013.

\bibitem{bouchaud2010price}
Jean-Philippe Bouchaud.
\newblock Price impact.
\newblock {\em Encyclopedia of quantitative finance}, 2010.

\bibitem{bouchaud2002statistical}
Jean-Philippe Bouchaud, Marc M{\'e}zard, and Marc Potters.
\newblock Statistical properties of stock order books: empirical results and
  models.
\newblock {\em Quantitative finance}, 2(4):251, 2002.

\bibitem{cartea2016closed}
{\'A}lvaro Cartea and Sebastian Jaimungal.
\newblock A closed-form execution strategy to target volume weighted average
  price.
\newblock {\em SIAM Journal on Financial Mathematics}, 7(1):760--785, 2016.

\bibitem{cont2014price}
Rama Cont, Arseniy Kukanov, and Sasha Stoikov.
\newblock The price impact of order book events.
\newblock {\em Journal of financial econometrics}, 12(1):47--88, 2014.

\bibitem{gould2013limit}
Martin~D Gould, Mason~A Porter, Stacy Williams, Mark McDonald, Daniel~J Fenn,
  and Sam~D Howison.
\newblock Limit order books.
\newblock {\em Quantitative Finance}, 13(11):1709--1742, 2013.

\bibitem{gu2008empirical}
Gao-Feng Gu, Wei Chen, and Wei-Xing Zhou.
\newblock Empirical regularities of order placement in the chinese stock
  market.
\newblock {\em Physica A: Statistical Mechanics and its Applications},
  387(13):3173--3182, 2008.

\bibitem{hendricks2014reinforcement}
Dieter Hendricks and Diane Wilcox.
\newblock A reinforcement learning extension to the almgren-chriss framework
  for optimal trade execution.
\newblock In {\em 2014 IEEE Conference on Computational Intelligence for
  Financial Engineering \& Economics (CIFEr)}, pages 457--464. IEEE, 2014.

\bibitem{hochreiter1998vanishing}
Sepp Hochreiter.
\newblock The vanishing gradient problem during learning recurrent neural nets
  and problem solutions.
\newblock {\em International Journal of Uncertainty, Fuzziness and
  Knowledge-Based Systems}, 6(02):107--116, 1998.

\bibitem{hochreiter1997long}
Sepp Hochreiter and J{\"u}rgen Schmidhuber.
\newblock Long short-term memory.
\newblock {\em Neural computation}, 9(8):1735--1780, 1997.

\bibitem{huberman2005optimal}
Gur Huberman and Werner Stanzl.
\newblock Optimal liquidity trading.
\newblock {\em Review of finance}, 9(2):165--200, 2005.

\bibitem{kato2015vwap}
Takashi Kato.
\newblock {VWAP} execution as an optimal strategy.
\newblock {\em JSIAM Letters}, 7:33--36, 2015.

\bibitem{lin2020end}
Siyu Lin and Peter~A Beling.
\newblock An end-to-end optimal trade execution framework based on proximal
  policy optimization.
\newblock In {\em IJCAI}, pages 4548--4554, 2020.

\bibitem{maskawa2007correlation}
Jun-ichi Maskawa.
\newblock Correlation of coming limit price with order book in stock markets.
\newblock {\em Physica A: Statistical Mechanics and its Applications},
  383(1):90--95, 2007.

\bibitem{nevmyvaka2006reinforcement}
Yuriy Nevmyvaka, Yi~Feng, and Michael Kearns.
\newblock Reinforcement learning for optimized trade execution.
\newblock In {\em Proceedings of the 23rd international conference on Machine
  learning}, pages 673--680, 2006.

\bibitem{ning2021double}
Brian Ning, Franco Ho~Ting Lin, and Sebastian Jaimungal.
\newblock Double deep {Q}-learning for optimal execution.
\newblock {\em Applied Mathematical Finance}, 28(4):361--380, 2021.

\bibitem{obizhaeva2013optimal}
Anna~A Obizhaeva and Jiang Wang.
\newblock Optimal trading strategy and supply/demand dynamics.
\newblock {\em Journal of Financial Markets}, 16(1):1--32, 2013.

\bibitem{platania2018modelling}
Federico Platania, Pedro Serrano, and Mikel Tapia.
\newblock Modelling the shape of the limit order book.
\newblock {\em Quantitative Finance}, 18(9):1575--1597, 2018.

\bibitem{potters2003more}
Marc Potters and Jean-Philippe Bouchaud.
\newblock More statistical properties of order books and price impact.
\newblock {\em Physica A: Statistical Mechanics and its Applications},
  324(1-2):133--140, 2003.

\bibitem{rogers2010cost}
Leonard~CG Rogers and Surbjeet Singh.
\newblock The cost of illiquidity and its effects on hedging.
\newblock {\em Mathematical Finance: An International Journal of Mathematics,
  Statistics and Financial Economics}, 20(4):597--615, 2010.

\bibitem{schnaubelt2022deep}
Matthias Schnaubelt.
\newblock Deep reinforcement learning for the optimal placement of
  cryptocurrency limit orders.
\newblock {\em European Journal of Operational Research}, 296(3):993--1006,
  2022.

\bibitem{weber2005order}
Philipp Weber and Bernd Rosenow.
\newblock Order book approach to price impact.
\newblock {\em Quantitative Finance}, 5(4):357--364, 2005.

\bibitem{ZGA21}
Mohammad Zainal, Ibrahim Gad, and Hameed AlQaheri.
\newblock An optimal limit order book prediction analysis based on deep
  learning and pigeon-inspired optimizer.
\newblock {\em Journal of system and management sciences}, 11(3):75--100, 2021.

\bibitem{zhang2019deeplob}
Zihao Zhang, Stefan Zohren, and Stephen Roberts.
\newblock Deeplob: Deep convolutional neural networks for limit order books.
\newblock {\em IEEE Transactions on Signal Processing}, 67(11):3001--3012,
  2019.

\bibitem{zovko2002power}
Ilija Zovko and J~Doyne Farmer.
\newblock The power of patience: a behavioural regularity in limit-order
  placement.
\newblock {\em Quantitative finance}, 2(5):387, 2002.

\end{thebibliography}

\bibliographystyle{plain}

\end{document}